\documentclass{article}
\usepackage[utf8]{inputenc}
\usepackage{authblk}

\usepackage{url}

\usepackage{breakurl}
\usepackage[breaklinks]{hyperref}

\title{Superlinear Lower Bounds Based on ETH}
\author[1]{Andr\'as Z. Salamon}
\author[2]{Michael Wehar}
\affil[1]{
  School of Computer Science, University of St Andrews, UK
  \href{mailto:Andras.Salamon@st-andrews.ac.uk}{Andras.Salamon@st-andrews.ac.uk}
}
\affil[2]{
  Computer Science Department, Swarthmore College, USA
  \href{mailto:mwehar1@swarthmore.edu}{mwehar1@swarthmore.edu}
}
\date{\today}

\usepackage{amsmath,amssymb,xspace}
\usepackage[textsize=small]{todonotes}

\usepackage{amsfonts}
\usepackage{amsthm}
\usepackage{color}

\newtheorem{proposition}{Proposition}
\newtheorem{lemma}[proposition]{Lemma}
\newtheorem{theorem}[proposition]{Theorem}
\newtheorem{corollary}[proposition]{Corollary}

\newtheorem{question}[proposition]{Question}
\newtheorem{definition}[proposition]{Definition}
\newtheorem{remark}[proposition]{Remark}
\newtheorem{example}[proposition]{Example}

\newcommand{\shortdash}{\hbox{-}}

\newcommand{\set}[2]{\{ \; #1 \; \mid \; #2 \; \}}

\newcommand{\gsum}[3]{\Sigma_{i=#1}^{#2}#3^i}

\newcommand{\poly}{\mathrm{poly}}
\newcommand{\polylog}{\mathrm{polylog}}

\newcommand{\cc}[1]{\ensuremath{\mathsf{#1}}}
\newcommand{\dtime}{\cc{DTIME}}
\newcommand{\dtiwi}{\cc{DTIWI}}

\newcommand{\ntigu}{\cc{NTIGU}}
\newcommand{\ntimeguess}{\cc{NTIMEGUESS}}

\newcommand{\ntime}{\ensuremath{\cc{NTIME}}\xspace}
\newcommand{\PTIME}{\ensuremath{\cc{P}}\xspace}
\newcommand{\NPTIME}{\ensuremath{\cc{NP}}\xspace}

\newcommand{\fpt}{\cc{FPT}}
\newcommand{\whier}{\cc{W}}
\newcommand{\WHier}[1]{\cc{W [ #1 ] }}

\newcommand{\sat}{\mathrm{SAT}}

\newcommand{\cnfsat}[1]{#1{\text -}\mathrm{CNF}{\text -}\mathrm{SAT}}
\newcommand{\circsat}{\mathrm{CircuitSAT}}
\newcommand{\logcircsat}{\mathrm{log\shortdash \mathrm{CircuitSAT}}}

\newcommand{\ovp}{\mathrm{OVP}}

\newcommand{\threesum}{\mathrm{3SUM}}
\newcommand{\trianglefinding}{\mathrm{Triangle \, Finding}}
\newcommand{\clique}{\mathrm{Clique}}

\newcommand{\apsp}{\mathrm{APSP}}
\newcommand{\eps}{\varepsilon}

\newcommand{\ethypo}{\mathrm{ETH}}
\newcommand{\sethypo}{\mathrm{SETH}}

\newcommand{\sub}{sub}
\newcommand{\SUB}[1]{\text{SUB}(#1)}
\newcommand{\Clo}[1]{\text{Clo}(#1)}
\newcommand{\len}[1]{\vert #1 \vert}
\newcommand{\pcount}[1]{\#_p ( #1 )}
\newcommand{\N}{\mathbb{N}}

\newcommand{\zap}[1]{}

\makeatletter
\def\blfootnote{\gdef\@thefnmark{}\@footnotetext}
\makeatother

\begin{document}

\maketitle

\begin{abstract}
We introduce techniques for proving superlinear conditional lower bounds for polynomial time problems.
In particular, we show that $\circsat$ for circuits with $m$ gates and $\log(m)$ inputs (denoted by $\logcircsat$) is not decidable in essentially-linear time unless the exponential time hypothesis ($\ethypo$) is false and $k$-$\clique$ is decidable in essentially-linear time in terms of the graph's size for all fixed $k$.
Such conditional lower bounds have previously only been demonstrated relative to the strong exponential time hypothesis (SETH).
Our results therefore offer significant progress towards proving unconditional superlinear time complexity lower bounds for natural problems in polynomial time.
\end{abstract}


\blfootnote{* Andr\'as Z. Salamon acknowledges support from EPSRC grants EP/P015638/1 and EP/V027182/1.}

\section{Introduction}
\subsection{Motivation}
Developing a deeper understanding of polynomial time problems is essential to the fields of algorithm design and computational complexity theory.
In this work, we build on prior concepts from the topic of limited nondeterminism to show a new kind of conditional lower bound for polynomial time problems where a small runtime improvement for one problem would lead to a substantial runtime improvement for another.

We proceed by introducing basic notions and explaining how they relate to existing work.
A polynomial time problem is a decision problem that can be decided in $O(n^k)$ time for some constant $k$, where $n$ denotes the input length.  As usual, \cc{P} denotes the class of polynomial time problems. A decision problem has an unconditional time complexity lower bound $t(n)$ if it cannot be decided in $o(t(n))$ time.  Polynomial time problems with unconditional superlinear time complexity lower bounds do not commonly appear in complexity theory research (aside from problems with lower bounds based on restrictive models such as one-tape Turing machines~\cite{Hennie1965:one-tape,Maass:1984}).  Such problems are known to exist by the deterministic time hierarchy theorem \cite{Hartmanis1965:computational}, but to the best of our knowledge, there are few examples that appear in the literature.  Most of the known examples are related to pebbling games \cite{Kasai:1984} or intersection non-emptiness for automata \cite{Wehar:2015,Wehar:2016}.  For these examples, the unconditional lower bounds are proven by combining Turing machine simulations with classical diagonalization arguments.

Although unconditional lower bounds are rare, many polynomial time problems have been shown to have conditional lower bounds in recent works on fine-grained complexity theory (see surveys \cite{Virginia:2015,Bringmann:2019}).
Our primary goal is to introduce superlinear conditional lower bounds based on weaker hypotheses than existing works, by applying new relationships between deterministic and nondeterministic computations.

\subsection{Our Contribution}
In this work, all logarithms are base 2 and we say that a problem is solvable in essentially-linear time if it is decidable in $O(n^{1+\eps})$ time
for all $\eps > 0$.

The $\logcircsat$ decision problem (previously investigated in \cite{Buss:1993,Abrahamson:1995}) is a natural restriction of circuit satisfiability to bounded fan-in Boolean circuits with $m$ gates and $\log(m)$ inputs.
Like many problems in polynomial time, it is not currently known if unconditional superlinear time complexity lower bounds exist for $\logcircsat$.
We prove a superlinear conditional lower bound for $\logcircsat$, as our main contribution.
(Our conditional lower bound is in fact superquasilinear, where $f$ is a quasilinear function if $f(n)/n \in \polylog(n)$.)
In particular, we show in Theorem~\ref{theorem:cond-lower-bound2} that $\logcircsat$ is not decidable in essentially-linear time unless the Exponential Time Hypothesis ($\ethypo$) is false.
This result is significant because existing works have only obtained conditional lower bounds for polynomial time problems based on the Strong Exponential Time Hypothesis ($\sethypo$).
It is well known that $\sethypo$ implies $\ethypo$ but the reverse implication is not known to hold~\cite[Theorem 14.5]{Cygan2015:parameterized}.
In fact, it has been claimed that while $\ethypo$ is plausible, $\sethypo$ ``is regarded by many as a quite doubtful working hypothesis that can be refuted at any time''~\cite[p.~470]{Cygan2015:parameterized}.
We therefore believe that a conditional lower bound for a natural polynomial time problem based on $\ethypo$ instead of $\sethypo$ represents significant progress.

As a further contribution, in Theorem~\ref{theorem:clique-speedup} we show that $\logcircsat$ is not decidable in essentially-linear time unless $k$-$\clique$ is decidable in essentially-linear time in terms of the graph's size for all fixed $k$.
This result is significant because the current best known algorithm for deciding $k$-$\clique$ runs in $O(v^{0.792 k})$ time \cite{Williams2009:clique} where $v$ denotes the number of vertices.
Furthermore, showing that there exists a constant $c$ such that \mbox{$k$-$\clique$} is decidable in $O(v^{c})$ time for all fixed $k$ would constitute a major breakthrough.

Our results for $\logcircsat$ follow from Speed-up Theorems~\ref{theorem:iterative1},~\ref{theorem:speedup}, and~\ref{theorem:parameterized-speedup}.  These theorems show how a small runtime improvement for the deterministic simulation of nondeterministic machines with short witnesses would imply a substantial runtime improvement for the deterministic simulation of nondeterministic machines with large witnesses.  Furthermore, these results advance our knowledge of limited nondeterminism by exploring possible trade-offs between time and witness length.

Our techniques are straightforward adaptations of existing approaches to simulation; our contribution is a more detailed analysis of these simulations and how they behave when composed and iterated.
This more detailed analysis is made possible by our novel approach to limited nondeterminism in Section~\ref{sec:newmodel}.

\section{Background}

Let $\N$ denote the set of positive integers $\{ 1, 2, \dots \}$.
The class of polynomial functions is
\[
  \poly(n) = \bigcup_{c>0} \set{ f(n) \colon \N \to \N }{ f(n) = O(n^c) },
\]
and in a slight abuse of notation, we also sometimes use $\poly(n)$ to mean an arbitrary function from this class.
We also refer to the class of polylogarithmic functions
\[
  \polylog(n) = \bigcup_{c>0} \set{ f(n) \colon \N \to \N }{ f(n) = O((\log( n))^c) }.
\]

\subsection{Conditional Lower Bounds}
Fine-grained complexity theory is a subject focused on exact runtime bounds and conditional lower bounds.
A conditional lower bound for a polynomial time problem typically takes the following form: polynomial time problem $A$ is not decidable in $O(n^{\alpha-\eps})$ time for all $\eps > 0$ assuming that problem $B$ is not decidable in $O(t(n)^{\beta-\eps})$ time for all $\eps > 0$, where $\alpha$ and $\beta$ are constants and $t(n)$ is a function (typically $t(n)$ is either polynomial or exponential).  This is referred to as a conditional lower bound because problem $A$ has a lower bound under the assumption that $B$ has a lower bound.  Conditional lower bounds are known for many polynomial time problems including \mbox{$\trianglefinding$}, Orthogonal Vectors Problem ($\ovp$), $\threesum$, and All Pairs Shortest Path ($\apsp$)~\cite{Virginia:2015,Bringmann:2019}.

\subsection{Exponential Time Hypothesis}

Decision problems related to Boolean formulas have been significant to the study of computational hardness \cite{Cook:1971,Levin:1973}.
As a result, satisfiability of Boolean formulas ($\sat$) is a natural candidate for lower bound assumptions.  In particular, it is common to focus on satisfiability of Boolean formulas in conjunctive normal form with clause width at most $k$ (denoted by $\cnfsat{k}$) for a fixed $k$.

The exponential time hypothesis ($\ethypo$) states that there is some $\eps > 0$, such that $\cnfsat{3}$ cannot be decided in $\poly(n) \cdot 2^{\eps \cdot v}$ time, where $n$ denotes the input size and $v$ denotes the number of variables \cite{ImpagliazzoOne:2001}.  The strong exponential time hypothesis ($\sethypo$) states that for every $\eps > 0$, there is a sufficiently large $k$ such that $\cnfsat{k}$ cannot be decided in $\poly(n) \cdot 2^{(1 - \eps) \cdot v}$ time \cite{ImpagliazzoOne:2001,ImpagliazzoTwo:2001,Impagliazzo:2009}.


Conditional lower bounds are frequently shown relative to $\cnfsat{k}$.  For instance, it is well known that the Orthogonal Vectors Problem ($\ovp$) on polylogarithmic length vectors is not decidable in $O(n^{2 - \eps})$ time for all $\eps > 0$ assuming $\sethypo$~\cite{Williams:2005,Virginia:2015,Bringmann:2019}.

\begin{remark}
As far as we know, the current best reduction shows that an $O(n^{\alpha})$ time algorithm for $\ovp$ would lead to a $\poly(n) \cdot 2^{\frac{\alpha \cdot v}{2}}$ time algorithm for $\cnfsat{k}$ for all $k$ \cite{Williams:2005,Virginia:2015,Bringmann:2019}.
This isn't sufficient to show a lower bound conditional on $\ethypo$.
Furthermore, we do not know if the existence of an essentially-linear time algorithm for $\ovp$ would imply that $\ethypo$ is false.
\end{remark}

\subsection{Limited Nondeterminism}

A nondeterministic polynomial time problem is a problem that can be decided in polynomial time by a nondeterministic machine.  Nondeterminism can appear in a computation in multiple different ways.  For instance, a machine could have nondeterministic bits written on a tape in advance, or could make nondeterministic guesses during the computation.  Such variations do not appear to make much difference for nondeterministic polynomial time ($\NPTIME$).  However, the definitions do require special care for notions of limited nondeterminism, which refers to the restriction or bounding of the amount of nondeterminism in a computation.

We proceed by reviewing some prior models of limited nondeterminism.
Consider a machine model consisting of a multitape Turing machine with a special guess tape; all the remaining tapes are standard.
In this model, the number of bits of nondeterminism used by the machine is the number of cells of the guess tape that are accessed by the machine during a computation, multiplied by the number of bits represented by each cell.
The contents of the guess tape are referred to as the \emph{witness}.

Kintala and Fischer~\cite{Kintala1977:computations,Kintala1980:refining} defined $\PTIME_{f(n)}$ as the class of languages that can be decided by a polynomial time bounded machine which scans at most $f(n)$ cells of the guess tape for each input of size $n$.
Note that this concept uses an exact limit for the amount of nondeterminism.
Abandoning this exactness,
\`{A}lvarez, D\'{\i}az and Tor\'an~\cite{Alvarez1989:complexity,Diaz1990:classes}, making explicit a concept of Xu, Doner, and Book~\cite{Xu1983:refining}, then defined $\beta_k$ as the class of languages that can be decided by a polynomial-time bounded machine which uses at most $O((\log(n))^k)$ bits of nondeterminism.
Farr took a similar approach~\cite{Farr1986:topics}, defining $f(n)$-$\NPTIME$ as the languages that can be decided by a polynomial-time bounded machine which scans at most $f(q(n))$ cells of the guess tape for each input of size $n$.
Here $q$ is a polynomial that depends only on the machine, so again there is an unspecified constant factor allowed in the amount of nondeterminism.
Note that $f(n)$-$\NPTIME$ is the union over all $k$ of the classes $\PTIME_{f(n^k)}$.
Another related approach was taken by Buss and Goldsmith~\cite{Buss:1993} where $\cc{N}^m \cc{P}_l$ is defined as the class of languages decided by nondeterministic machines in quasi-$n^l$ time making at most $m \cdot \log(n)$ nondeterministic guesses.
In this approach the limit on the amount of nondeterminism is exact, but arbitrary poly-logarithmic factors are allowed in the time bound.
Finally, in the survey by Goldsmith, Levy and Mundhenk~\cite{Goldsmith1996:limited} the $\beta_k$ classes were then extended to verifiers other than those with a polynomial time bound.
In this notation, $\beta_k\cc{-C}$ is defined relative to a complexity class $\cc{C}$ that bounds the power of the verifier.  Therefore, we have $\beta_k = \beta_k\cc{-\PTIME}$.

Taking a slightly different approach, Cai and Chen~\cite{Cai1997:amount} focused on machines that partition access to nondeterminism, by first creating the contents of the guess tape, and then using a deterministic machine to check this guess.
In this terminology, $\cc{GC}(s(n),\cc{C})$ is the class of languages that can be decided by a machine that guesses $O(s(n))$ bits and then uses the power of class $\cc{C}$ to verify.
Again an arbitrary constant factor is allowed in the number of nondeterministic bits, to allow classes to contain complete languages.

Santhanam~\cite{Santhanam:2000} then returned to a definition that uses an exact limit for the amount of allowed nondeterminism: $\ntigu(t(n), g(n))$ is the class of languages that can be decided by a machine that makes $g(n)$ guesses and runs for $O(t(n))$ time.
These classes have also been denoted $\ntimeguess(t(n), g(n))$ in a more recent work~\cite{fortnow:2016}.

It follows from the definitions that
\[
  \PTIME = \PTIME_{O(\log(n))} = \ntigu(\poly(n), O(\log(n))) = \cc{GC}(O(1),\PTIME) = \cc{N}^{O(1)} \cc{P}_{O(1)} = \beta_1 
\]
and
\[
  \NPTIME = n{\text -}\NPTIME  = \PTIME_{n^{O(1)}} = \ntigu(\poly(n), \poly(n)) = \cc{GC}(n^{O(1)},\PTIME).
\]
Furthermore, the $\beta_k$ classes are meant to capture classes between $\PTIME$ and $\NPTIME$.

\begin{remark}
In this work, we focus on the $\logcircsat$ problem and the levels within $\PTIME = \beta_1$.  It is worth noting that a loosely related work \cite{Feige1997:limited} investigated the $\log$-$\clique$ problem which is in $\beta_2 = \ntigu(\poly(n),O(\log(n)^2))$.
\end{remark}






\section{Time-Witness Trade-offs}

In the following, we introduce a new notion of limited nondeterminism that we use to prove new relationships between deterministic and nondeterministic computations.  In particular, we prove that if faster deterministic algorithms exist, then there are straightforward trade-offs between time and witness length.
For our notion of limited nondeterminism, unlike existing models, the nondeterministic guesses are preallocated as placeholders within an input string.  These placeholders can then be filled with nondeterministic bits.
It is important to note that different models of limited nondeterminism could be used.  However, our model allows us to preserve the input size enabling us to prove technical results, Lemmas~\ref{lemma:witness} and \ref{lemma:padding}.
Attempting to prove a result like Lemma~\ref{lemma:witness} for a model such as $\ntigu$ introduces unnecessary challenges with managing input and guess strings.

\subsection{A New Model for Limited Nondeterminism}
\label{sec:newmodel}

The attempts at constructing robust classes containing complete problems by allowing arbitrary factors in the amount of nondeterminism were challenged by the various downward collapses of the $\beta$ hierarchy shown by Beigel and Goldsmith relative to oracles~\cite{Beigel1998:downward}.
We therefore need a notion of limited nondeterminism that tracks constant factors in the amount of nondeterminism, accepting a lack of complete problems in our complexity classes to gain greater precision in reductions.
This suggests using the $\ntigu$ notation.

However, we found that attempting to use the $\ntigu$ notion directly leads to difficulties with bookkeeping when composing multiple reductions because of the necessary simultaneous management of input and guess strings.
Since composing reductions is at the heart of our approach for proving speed-up theorems in Subsection~\ref{subsection:speedup},
we sought a different notion that overcomes these unnecessary technical obstacles.

We now introduce our model of limited nondeterminism which allows us to be more explicit than the $\ntigu$ classes in keeping track of the witness bits when composing reductions.
Reminiscent of the Cai and Chen guess-and-check classes, in our model the nondeterministic guesses will be preallocated as placeholder characters within an input string.  This means that we can only fill in placeholder characters with nondeterministic bits.  This property is essential for proving structural properties (see the translation and padding lemmas in Subsection~\ref{subsection:structural-properties}).  With other models, proofs of structural properties appear to be intrinsically more complex, requiring separate treatment of various overheads and applications of tape reduction theorems.

Consider strings over a ternary alphabet $\Sigma = \{ 0, 1, p \}$ where $p$ is referred to as the placeholder character.
We index the bits of a string starting from position $0$.
For any string $x \in \Sigma^{*}$, we let $\len{x}$ denote the length of $x$ and $\pcount{x}$ denote the number of placeholder character occurrences in $x$.

\begin{definition}
Let a string $r \in \{0, 1\}^{*}$ be given.  Define a function
\[
  \sub_{r}: \Sigma^{*} \rightarrow \Sigma^{*}
\]
such that for each string $x \in \Sigma^{*}$, $\sub_{r}(x)$ is obtained from $x$ by replacing placeholder characters with bits from $r$ so that the $i$th placeholder character from $x$ is replaced by the $i$th bit of $r$ for all $i$ satisfying $0 \leq i < min\{\len{r}, \pcount{x}\}$.  Also, define $\SUB{n} := \set{\sub_{r}}{\len{r} \leq n}$.  We call $\sub_r$ a prefix filling and $\SUB{n}$ a set of prefix fillings.
\end{definition}

\begin{example}
Consider strings $x = 11p01p0p$ and $r = 0110$.  By applying the preceding definition, we have that $\sub_{r}(x) = 11001101$.
\end{example}

A prefix filling $\sub_r$ replaces the first $\len{r}$ placeholder characters with the bits of $r$ in order (from the least index to the greatest).
If there are fewer placeholder characters then some of the bits of $r$ remain unused.
We also consider the notion of an \emph{unrestricted filling}, which is any injective replacement of placeholder characters, without specifying the particular order.

\begin{definition}
For strings $x$ and $y \in \Sigma^{*}$, we write $x \preceq y$ if $x$ can be obtained from $y$ by replacing any number of placeholder characters in $y$ with 0 or 1.
Given a language $L \subseteq \Sigma^{*}$, we let the closure of $L$ under unrestricted fillings be
\[
  \Clo{L} := \set{x \in \Sigma^{*}}{(\exists \; y \in L) \; x \preceq y}.
\]
\end{definition}

\begin{example}
Consider a language $L = \{ 0p1p \}$.  By applying the preceding definition, we have that $\Clo{L} = \{ 0p1p, 0p10, 0p11, 001p, 011p, 0010, 0011, 0110, 0111 \}$.
\end{example}

In the following, $\dtime(t(n))$ represents the class of languages decidable in $O(t(n))$ time by multitape Turing machines that have read and write access to all tapes (including the input tape).
We proceed by defining a complexity class $\dtiwi(t(n), w(n))$ where intuitively $t(n)$ represents a time bound and $w(n)$ represents a bound on witness length.

\begin{definition}
Let $\Sigma = \{0,1,p\}$ and consider a language $L \subseteq \Sigma^{*}$.  We write
\[
  L \in \dtiwi(t(n), w(n))
\]
if there exist languages $U$ and $V \subseteq \Sigma^{*}$ satisfying the following properties:
\begin{itemize}
\item $U \in \dtime(n)$,
\item $\Clo{U} \in \dtime(n)$,
\item $V \in \dtime(t(n))$, and
\item for all $x \in \Sigma^{*}$, $x \in L$ if and only if $x \in U$ and there exists \mbox{$s \in \SUB{w(\len{x})}$} such that $s(x) \in V$.
\end{itemize}
We refer to $V$ as a verification language for $L$ with input string universe $U$.
\end{definition}



There are many different ways to encode structures as strings over a fixed alphabet, so decision problems can take many different forms as formal languages.  To put a problem within $\dtiwi(t(n), w(n))$, we therefore need to provide an encoding for its inputs with placeholder characters at the appropriate positions.



\begin{example}\label{example:sat}
$\sat$ can be represented such that each input has placeholder characters out front followed by an encoding of a Boolean formula.  Each variable is represented as a binary number representing an index to a placeholder.  The placeholders will be nondeterministically filled to create a variable assignment.
\end{example}


\subsection{Structural Properties of Limited Nondeterminism}\label{subsection:structural-properties}
The following two lemmas demonstrate structural properties relating time and witness length.  These properties will be essential to proving speed-up theorems in Subsection~\ref{subsection:speedup} that reveal new relationships between deterministic and nondeterministic computations.

\begin{lemma}[Translation Lemma]\label{lemma:witness}
If $\dtiwi(t(n),w(n)) \subseteq \dtime(t'(n))$, then for all $w'$,
\[
  \dtiwi(t(n),w(n) + w'(n)) \subseteq \dtiwi(t'(n), w'(n)).
\]
\end{lemma}

\begin{proof}
Suppose that $\dtiwi(t(n),w(n)) \subseteq \dtime(t'(n))$.

Let a function $w'$ be given.  Let $L \in \dtiwi(t(n),w(n) + w'(n))$ be given.  By definition, there exist an input string universe $U$ and a verification language $V \in \dtime(t(n))$ satisfying that $\forall x \in \Sigma^{*}$, $x \in L$ if and only if $x \in U$ and there exists $s \in \SUB{w(\len{x}) + w'(\len{x})}$ such that $s(x) \in V$.  Consider a new language
\[
  L' := \set{x \in \Clo{U}}{(\exists s \in \SUB{w(\len{x})}) \; s(x) \in V}.
\]
 %
By interpreting $V$ as a verification language for $L'$ with input string universe $\Clo{U}$, we get $L' \in \dtiwi(t(n), w(n))$.  By assumption, it follows that
\[
  L' \in \dtime(t'(n)).
\]
Finally, by interpreting $L'$ as a verification language for $L$ with input string universe $U$, we get $L \in \dtiwi(t'(n), w'(n))$.
\end{proof}

Recall that a function $f \colon \N \to \N$ is \emph{fully time-constructible} if there is a deterministic multitape Turing machine $M$ that for every input of length $n$ runs for exactly $f(n)$ steps~\cite{Homer2011:computability}.
By convention, if $f(n)$ is a fully time-constructible function, then $f(n) \ge n$ for all $n \in \N$.

\begin{lemma}[Padding Lemma]\label{lemma:padding}
If $\dtiwi(t(n),w(n)) \subseteq \dtime(t'(n))$, then for all fully time-constructible functions $f$,
\[
  \dtiwi(t(f(n)),w(f(n))) \subseteq \dtime(t'(f(n))).
\]
\end{lemma}

\begin{proof}
Suppose that $\dtiwi(t(n),w(n)) \subseteq \dtime(t'(n))$,
and that $f \colon \N \to \N$ is fully time-constructible.
By definition, there exist an input string universe $U$ and a verification language
\[
  V \in \dtime(t(f(n)))
\]
so that $\forall x \in \Sigma^{*}$, $x \in L$ if and only if $x \in U$ and there exists
\[
  s \in \SUB{w(f(\len{x}))}
\]
such that $s(x) \in V$.
Consider new languages $L'$, $V'$, and $U'$ such that
\begin{align*}
L' &:= \set{1^{k-1} \cdot 0 \cdot x}{k + \len{x} = f(\len{x}) \; \wedge \; x \in L}, \\
V' &:= \set{1^{k-1} \cdot 0 \cdot x}{k + \len{x} = f(\len{x}) \; \wedge \; x \in V}, \text{ and}\\
U' &:= \set{1^{k-1} \cdot 0 \cdot x}{k \geq 1 \; \wedge \; x \in U}.
\end{align*}
Since $V \in \dtime(t(f(n)))$,
we have that $V' \in \dtime(t(n))$.  By interpreting $V'$ as a verification language for $L'$ with input string universe $U'$, we get $L' \in \dtiwi(t(n), w(n))$.  By assumption, it follows that $L' \in \dtime(t'(n))$.
We conclude that $L \in \dtime(t'(f(n)))$.
\end{proof}

\begin{remark}\label{remark:dtiwi-explanation}
Initially, we tried to use other notions of limited nondeterminism such as $\ntigu$ to prove the preceding lemmas.
However, the proofs were messy and required increasing the number of Turing machine tapes or the time complexity.
In contrast, our model for limited nondeterminism ($\dtiwi$) preserves the input size leading to straightforward proofs with tighter complexity bounds.
\end{remark}

\subsection{Speed-up Theorems}\label{subsection:speedup}
In this subsection, we carefully prove three speed-up theorems relating time and witness length.  It is important to mention that there are existing speed-up theorems in the recent literature relating different computational resources such as those relating time and space in \mbox{\cite{Williams:2013,Buss:2014}} and relating probabilistic circuit size and success probability in \cite{Paturi2010:complexity}.  In addition, although relevant, we note that our speed-up results are distinct from recent hardness magnification results \cite{Williams:2020} which amplify circuit lower bounds rather than speed up computations.


The first speed-up theorem follows by repeatedly applying the structural properties
of limited nondeterminism from the preceding subsection.

\begin{theorem}[First Speed-up Theorem]\label{theorem:iterative1}
Let $\alpha$ be a rational number such that \mbox{$1 \leq \alpha < 2$.
If}
\[
  \dtiwi(n,\log(n)) \subseteq \dtime(n^{\alpha}),
\]
then for all $k \in \mathbb{N}$, $\dtiwi(n,(\gsum{0}{k}{\alpha})\log(n)) \subseteq \dtime(n^{\alpha^{k + 1}})$.
\end{theorem}

\begin{proof}

Suppose $\alpha$ is rational and $1 \leq \alpha < 2$.
(Note that when $\alpha = 1$ some of the following formulas can be simplified, but the proof still holds for this case.)

Now suppose that $\dtiwi(n,\log(n)) \subseteq \dtime(n^{\alpha})$.  We prove by induction on $k$ that for all $k \in \mathbb{N}$,
\[
  \dtiwi(n,(\gsum{0}{k}{\alpha})\log(n)) \subseteq \dtime(n^{\alpha^{k + 1}}).
\]
The base case ($k = 0$) is true by assumption.  For the induction step, \mbox{suppose that}
\[
  \dtiwi(n,(\gsum{0}{k}{\alpha})\log(n)) \subseteq \dtime(n^{\alpha^{k + 1}}).
\]
By applying this assumption with Lemma~\ref{lemma:witness}, we get that
\[
  \dtiwi(n,(\gsum{0}{k+1}{\alpha})\log(n)) \subseteq \dtiwi(n^{\alpha^{k + 1}}, \alpha^{k+1} \cdot \log(n)).
\]
Let $f(n) = n^{\alpha^{k+1}}$, which is fully time-constructible~\cite[Example 1]{Kobayashi1985:proving}.
Next, we apply our initial assumption and Lemma~\ref{lemma:padding} with $f(n)$, $w(n) = \log(n)$, $t(n) = n$, and $t'(n) = n^{\alpha}$.
Therefore,
\[
  \dtiwi(n^{\alpha^{k + 1}}, \alpha^{k+1} \cdot \log(n)) \subseteq \dtime(n^{\alpha^{k + 2}}).
\]
It follows that $\dtiwi(n,(\gsum{0}{k+1}{\alpha})\log(n)) \subseteq \dtime(n^{\alpha^{k + 2}}).$
\end{proof}



\begin{remark}
Theorem~\ref{theorem:iterative1} is a speed-up result because when $1 \leq \alpha < 2$, the exponent from the runtime divided by the constant factor for the witness string length decreases as $k$ increases.  In particular, we have
\[
  \lim_{k\to\infty}\frac{\alpha^{k+1}}{\gsum{0}{k}{\alpha}} = (\alpha - 1) \cdot \lim_{k\to\infty}\frac{\alpha^{k+1}}{\alpha^{k+1} - 1} = \alpha - 1 < 1.
\]
\end{remark}


The second speed-up theorem follows by combining the first speed-up theorem with the padding lemma.
We say that a function $g \colon \N \to \N$ is \emph{well-computable} if $g(n) \le n$ for every $n \in \N$, $g(n) = \omega(\log(n))$, and $g(n)$ can be computed in $\poly(n)$ steps.

\begin{theorem}[Second Speed-up Theorem]\label{theorem:speedup}
Suppose that $g$ is a well-computable function.
Let $\alpha$ be a rational number such that $1 < \alpha < 2$.
If
\[
  \dtiwi(n,\log(n)) \subseteq \dtime(n^{\alpha}),
\]
then
\[
  (\forall \eps > 0) \; \dtiwi(\poly(n), g(n)) \subseteq \dtime(2^{(1 + \eps) \cdot (\alpha - 1) \cdot g(n)}).
\]
\end{theorem}

\begin{proof}
Let $\alpha$ be a rational number such that $1 < \alpha < 2$.
Let $z(\alpha,k) = \gsum{0}{k}{\alpha}$.
Note that
\[
  z(\alpha,k) = \frac{\alpha^{k+1} - 1}{\alpha - 1}.
\]

Suppose that $\dtiwi(n,\log(n)) \subseteq \dtime(n^{\alpha})$.  Let $\eps > 0$ be given.  By Theorem~\ref{theorem:iterative1}, we have that for all $k \in \mathbb{N}$,
\[
  \dtiwi(n,z(\alpha,k)\log(n)) \subseteq \dtime(n^{\alpha^{k + 1}}).
\]

Let $f(n) = 2^{\lceil g(n)/z(\alpha,k) \rceil}$ if $g(n) \ge z(\alpha,k)\log(n)$ and $f(n) = n$ otherwise.
Because $g(n) = \omega(\log(n))$, there is some $c > 1$ such that $f(n) > cn$ for all but finitely many $n \in \N$.
Now, since $g(n)$ can be computed in $\poly(n)$ time and $z(\alpha,k)$ is rational, $f(n)$ can be computed in binary in $\poly(n)$ time.  Furthermore, since $f(n)$ is superpolynomial, $f(n)$ can be computed in $O(f(n))$ time.
Therefore, by \cite[Theorem 4.1]{Kobayashi1985:proving}, $f(n)$ is fully time-constructible.
Next, we apply Lemma~\ref{lemma:padding} with $f(n)$, \mbox{$w(n) = z(\alpha,k)\log(n)$}, and $t(n) = n$.  Therefore
\[
  \dtiwi(2^{g(n)/z(\alpha,k)},g(n)) \subseteq \dtime(2^{(\alpha^{k+1}) \cdot g(n)/z(\alpha,k)}).
\]
Again, since $2^{g(n)/z(\alpha,k)}$ is superpolynomial, we have
\[
  \dtiwi(\poly(n),g(n)) \subseteq \dtime(2^{(\alpha^{k+1}) \cdot g(n)/z(\alpha,k)}).
\]
Then, since
\[
  \lim_{k\to\infty} \frac{\alpha^{k + 1}}{\alpha^{k+1} - 1} = 1,
\]
there exists $k$ sufficiently large such that
\[
  \frac{\alpha^{k+1}}{\alpha^{k+1} - 1} \leq 1 + \eps.
\]
Therefore, by choosing sufficiently large $k$, we have
\[
  \dtiwi(\poly(n),g(n)) \subseteq \dtime(2^{(1 + \eps) \cdot (\alpha - 1) \cdot g(n)}).\qedhere
\]
\end{proof}

\begin{corollary}\label{corollary:speedup}
Suppose that $g$ is a well-computable function.
If for all $\alpha > 1$,
\[
  \dtiwi(n, \log(n)) \subseteq \dtime(n^{\alpha}),
\]
then
$(\forall \eps > 0) \; \dtiwi(\poly(n),g(n)) \subseteq \dtime(2^{\eps \cdot g(n)}).$
\end{corollary}

\begin{proof}
Follows directly from Theorem~\ref{theorem:speedup}.
\end{proof}

The third speed-up theorem follows by carefully applying the first speed-up theorem.

\begin{theorem}[Third Speed-up Theorem]\label{theorem:parameterized-speedup}
If for all $\alpha > 1$,
\[
  \dtiwi(n, \log(n)) \subseteq \dtime(n^{\alpha}),
\]
then for all $k \in \N$ and all $\alpha > 1$,
\(
  \dtiwi(n, k \cdot \log(n)) \subseteq \dtime(n^{\alpha}).
\)
\end{theorem}

\begin{proof}
Suppose that for all $\alpha > 1$, $\dtiwi(n, \log(n)) \subseteq \dtime(n^{\alpha})$.  By Theorem~\ref{theorem:iterative1}, for all rational $\alpha$ such that $1 < \alpha < 2$ and for all $k \in \mathbb{N}$,
\[
  \dtiwi(n,(\gsum{0}{k}{\alpha})\log(n)) \subseteq \dtime(n^{\alpha^{k + 1}}).
\]
Notice that when $\alpha > 1$, we have $k < (\gsum{0}{k}{\alpha})$.  Therefore, for all rational $\alpha$ such that $1 < \alpha < 2$ and for all $k \in \mathbb{N}$,
\[
  \dtiwi(n, k \cdot \log(n)) \subseteq \dtime(n^{\alpha^{k + 1}}).
\]

Now, let $k \in \mathbb{N}$ and a rational number $\alpha_1 > 1$ be given.  We can choose a rational number $\alpha_2 > 1$ sufficiently close to $1$ so that $\alpha_2^{k + 1} \leq \alpha_1$.  It follows that
\[
  \dtiwi(n, k \cdot \log(n)) \subseteq \dtime(n^{\alpha_2^{k + 1}}) \subseteq \dtime(n^{\alpha_1}).
\]
As the rationals form a dense subset of the reals, the result follows.
\end{proof}

%

\section{Superlinear Conditional Lower Bounds}
\label{section:stronger}
\subsection{log-CircuitSAT Decision Problem}

A common generalization of $\sat$ is the problem of deciding satisfiability of Boolean circuits (denoted by $\circsat$).  There is a natural restriction of $\circsat$ to bounded fan-in Boolean circuits with $m$ gates and $\log(m)$ inputs (denoted by $\logcircsat$) \cite{Buss:1993,Abrahamson:1995}.
We encode this problem so that the placeholder characters are out front followed by an encoding of a bounded fan-in Boolean circuit.  Such an encoding can be carried out so that if $n$ denotes the total input length and $m$ denotes the number of gates, then $n = \Theta(m \cdot \log(m))$.

The $\logcircsat$ decision problem is decidable in polynomial time because we can evaluate the circuit on every possible input assignment.  Whether or not we can decide $\logcircsat$ in $O(n^{2-\eps})$ time for some $\eps > 0$ is an open problem.  Furthermore, as far as we know, no unconditional superlinear lower bounds are known for $\logcircsat$.  Later in this section, we prove a superlinear conditional lower bound for $\logcircsat$.  In particular, we show that if $\logcircsat$ is decidable in essentially-linear time, then $\ethypo$ is false (Theorem~\ref{theorem:cond-lower-bound2}) meaning that a small runtime improvement for $\logcircsat$ would lead to a substantial runtime improvement for $\NPTIME$-complete problems.

\subsection{Simulating Turing Machines Using Boolean Circuits}\label{subsection:simulating-machines-using-circuits}
Let a fully time-constructible function $t$ be given.  Any $O(t(n))$ time bounded Turing machine can be simulated by an oblivious Turing machine in $O(t(n) \cdot \log(t(n)))$ time \cite{Pippenger:1979}.  Moreover, any $O(t(n))$ time bounded Turing machine can be simulated by Boolean circuits of size \mbox{$O(t(n) \cdot \log(t(n)))$} which can be computed efficiently by a Turing machine \cite{Cook1988:short}.

\begin{theorem}[\cite{Pippenger:1979,Cook1988:short,Buss:1993,Lipton:2012}]\label{theorem:circuits}
Let a fully time-constructible function $t$ be given.  If $L \in \dtime(t(n))$, then in
\[
  O(t(n) \cdot \poly(\log(t(n))))
\]
time, we can compute Boolean circuits for $L$ of size at most $O(t(n) \cdot \log(t(n)))$.
\end{theorem}

We now use Theorem~\ref{theorem:circuits} to show that any problem in $\dtiwi(n,\log(n))$ is efficiently reducible to $\logcircsat$.

\begin{theorem}\label{theorem:hardness}
Any $L \in \dtiwi(n,\log(n))$ is reducible to logarithmically many instances of $\logcircsat$ in essentially-linear time by a Turing machine.
\end{theorem}

\begin{proof}
Let $L \in \dtiwi(n,\log(n))$ be given.  Let $V \in \dtime(n)$ denote a verification language for $L$ with input string universe $U \in \dtime(n)$.


Let an input string $x \in U$ of length $n$ be given.  By Theorem~\ref{theorem:circuits}, we can compute
a Boolean circuit\footnote{Since $L$ and $V$ are over a ternary alphabet, the input strings are encoded into binary before being fed into the Boolean circuits.} $C$ for $V$ with at most $O(n \cdot \log(n))$ gates in essentially-linear time on a Turing machine.  In the following, let $[\log(n)]$ denote $\{1,2,\dots,\lfloor \log(n) \rfloor\}$.  Now, we construct a family of circuits $\{ C_i \}_{[\log(n)]}$ such that for each $i \in [\log(n)]$, $C_i$ is obtained by fixing the characters of $x$ into the circuit $C$ so that only $i$ input bits remain where these input bits are associated with the first $i$ placeholders within $x$.  Therefore the circuit $C_i$ has at most $\log(n)$ inputs and at most $O(n \cdot \log(n))$ gates.  It follows that $x \in L$ if and only if there exists $i \in [\log(n)]$ such that $C_i$ is satisfiable.
\end{proof}





\begin{corollary}\label{corollary:hardness1}
If for all $\alpha > 1$ we have $\logcircsat \in \dtime(n^{\alpha})$, then for all $\alpha > 1$
\[
  \dtiwi(n, \log(n)) \subseteq \dtime(n^{\alpha}).
\]
\end{corollary}

\begin{proof}
Follows directly from Theorem~\ref{theorem:hardness}.
\end{proof}

\subsection{ETH-hardness}
We combine results from Subsection~\ref{subsection:simulating-machines-using-circuits} with the Second Speed-up Theorem to prove superlinear conditional lower bounds for $\logcircsat$.  In particular, existence of essentially-linear time algorithms for $\logcircsat$ would imply that $\ethypo$ is false.

\begin{corollary}\label{corollary:hardness2}
Suppose that $g$ is a well-computable function.
If for all $\alpha > 1$,
\[
  \logcircsat \in \dtime(n^{\alpha}),
\]
then
$(\forall \eps > 0) \; \dtiwi(\poly(n),g(n)) \subseteq \dtime(2^{\eps \cdot g(n)}).$
\end{corollary}

\begin{proof}
Follows by combining Corollary~\ref{corollary:hardness1} with Corollary~\ref{corollary:speedup}.
\end{proof}

We now relate $\logcircsat$ and $\circsat$, showing that an essentially-linear upper bound for $\logcircsat$ would imply a subexponential upper bound for $\circsat$.

\begin{theorem}\label{theorem:cond-lower-bound1}
If for every $\alpha > 1$ we have that $\logcircsat \in \dtime(n^{\alpha})$, then
\[
  (\forall \eps > 0) \; \circsat \in \dtime(\poly(n) \cdot 2^{\eps \cdot m}),
\]
where $m$ is the number of gates.
\end{theorem}

\begin{proof}
Suppose that for all $\alpha > 1$, $\logcircsat \in \dtime(n^\alpha)$.  Letting $\lg x = \max\{1,\log(x)\}$ and applying Corollary~\ref{corollary:hardness2} with $g(n) = \frac{n}{\lg(n)}$ (which is well-computable), we conclude that
\[
  (\forall \eps > 0)\; \dtiwi(\poly(n), \frac{n}{\log(n)}) \subseteq \dtime(2^{\eps \cdot \frac{n}{\log(n)}}).
\]

Recall that we encode Boolean circuits so that $n = \Theta(m \cdot \log(m))$ where $m$ is the number of gates.  Therefore, $\frac{n}{\log(n)}$ is $\Theta(m)$.  Also, under reasonable encoding conventions, $\frac{n}{\log(n)}$ will actually be larger than the number of gates and inputs. (Note that one could always scale up the witness size by a rational constant factor if needed.) Hence,
\[
  \circsat \in \dtiwi(\poly(n), \frac{n}{\log(n)}).
\]
Therefore, $(\forall \eps > 0)\; \circsat \in \dtime(2^{\eps \cdot \frac{n}{\log(n)}})$.  It follows that
\[
  (\forall \eps > 0)\; \circsat \in \dtime(\poly(n) \cdot 2^{\eps \cdot m}).\qedhere
\]
\end{proof}

We now show that $\logcircsat$ cannot be decided in essentially-linear time unless $\ethypo$ fails.
Note that this is a conditional lower bound based on $\ethypo$ rather than the more common (and stronger) $\sethypo$ assumption.

\begin{theorem}\label{theorem:cond-lower-bound2}
If $\logcircsat \in \dtime(n^{\alpha})$ for every $\alpha > 1$, then $\ethypo$ is false.
\end{theorem}

\begin{proof}
Because $\cnfsat{3}$ is a special case of $\circsat$, Theorem~\ref{theorem:cond-lower-bound1} implies that
\[
  (\forall \eps > 0) \; \cnfsat{3} \in \dtime(\poly(n) \cdot 2^{\eps \cdot m})
\]
where $m$ is the number of bounded AND, OR, and NOT gates (which is approximately three times the number of clauses).  By applying the Sparsification Lemma \cite{ImpagliazzoTwo:2001,Lokshtanov2011:eth}, we get that
\[
  (\forall \eps > 0) \; \cnfsat{3} \in \dtime(\poly(n) \cdot 2^{\eps \cdot v})
\]
where $v$ is the number of variables.  It follows that $\ethypo$ is false.
\end{proof}

\subsection{Hardness for k-Clique}
We combine results from Subsection~\ref{subsection:simulating-machines-using-circuits} with the Third Speed-up Theorem to prove that the existence of essentially-linear time algorithms for $\logcircsat$ would imply that $k$-$\clique$ has essentially-linear time algorithms for all fixed $k$.
It is important to note that our construction is non-uniform meaning that we obtain differing algorithms that cannot necessarily be combined into a single efficient approach for solving $k$-$\clique$ on non-constant $k$.  As a result, our argument is not sufficient to conclude $\fpt = \WHier{1}$.

\begin{corollary}\label{corollary:parameterized-speedup}
If for every $\alpha > 1$ we have that $\logcircsat \in \dtime(n^{\alpha})$, then
\[
  \dtiwi(n, k \cdot \log(n)) \subseteq \dtime(n^{\alpha})
\]
for every $k \in \N$ and every $\alpha > 1$.
\end{corollary}

\begin{proof}
Follows by combining Corollary~\ref{corollary:hardness1} with Theorem~\ref{theorem:parameterized-speedup}.
\end{proof}

From this, we are able to obtain a meaningful connection between the $\logcircsat$ and $k{\text -}\clique$ problems.

\begin{theorem}\label{theorem:clique-speedup}
If for every $\alpha > 1$ we have that $\logcircsat \in \dtime(n^{\alpha})$, then
\[
  k{\text -}\clique \in \dtime(n^{\alpha})
\]
for every $k \in \N$ and every $\alpha > 1$.
\end{theorem}

\begin{proof}
The variable $n$ denotes the total length of the graph's encoding which is $\Theta((v + e) \cdot \log(v))$ where $v$ is the number of vertices and $e$ is the number of edges.
We observe that for all fixed $k$, we have
\[
  k{\text -}\clique \in \dtiwi(n, k \cdot \log(n)).
\]
We combine this observation with Corollary~\ref{corollary:parameterized-speedup} to obtain the desired result.
\end{proof}


\begin{remark}
Although we do not focus on parameterized complexity theory here, the preceding arguments can also be used to show that if $\logcircsat$ is decidable in essentially-linear time, then $\WHier{1} \subseteq \text{non-uniform-}\fpt$.  Moreover, we suggest that this implication could be extended to $\WHier{\PTIME} \subseteq \text{non-uniform-}\fpt$.  We refer the reader to \cite{Abrahamson:1995} for background on $\WHier{\PTIME}$ and the $\whier$ hierarchy.
\end{remark}

%
%

\section{Conclusion}
We have demonstrated superlinear conditional lower bounds for the $\logcircsat$ decision problem by carefully investigating properties of limited nondeterminism.  In particular, in Theorem~\ref{theorem:cond-lower-bound2} we showed that the existence of essentially-linear time Turing machines for $\logcircsat$ would imply that $\ethypo$ is false.  This means that a small runtime improvement for $\logcircsat$ would lead to a substantial runtime improvement for $\NPTIME$-complete problems.  Through this investigation we revealed new relationships between deterministic and nondeterministic computations.

We leave two important questions unanswered that we hope will inspire future work.

\begin{question}

Would the existence of essentially-linear time random access machines for $\logcircsat$ imply that $\ethypo$ is false?  This question is related to whether linear time for random access machines can be simulated in subquadratic time by multitape Turing machines~\cite{Cook:1973}.  It is also related to whether random access machines can be made oblivious~\cite{Gurevich1989:nearly}.

\end{question}


\begin{question}

Can the construction from the first speed-up theorem (Theorem~\ref{theorem:iterative1}) be carried out for a non-constant number $k$ of iterations?  We speculate that if it can, then $\dtiwi(n, \log(n)) \subseteq \dtime(n \cdot \log(n))$ would imply that $\ntime(n) \subseteq \dtime(2^{\sqrt{n}})$.

\end{question}

In addition, although this work does not focus on circuit lower bounds, we suggest that recent results connecting the existence of faster algorithms with circuit lower bounds \cite{Abboud:2016,Williams2014:nonuniform,Williams:2013,Viola:2014} could be applied to show that the existence of faster algorithms for $\logcircsat$ would imply new circuit lower bounds for $\cc{E}^{\NPTIME}$ as well as other complexity classes.

Finally, we leave the reader with the thought that the speed-up theorems for limited nondeterminism (Theorems~\ref{theorem:iterative1},~\ref{theorem:speedup}, and~\ref{theorem:parameterized-speedup}) might be special cases of a more general speed-up result connecting nondeterminism, alternation, and time.

\paragraph{Acknowledgements}
We greatly appreciate the help and suggestions that we received.  We are especially grateful to Kenneth Regan and Jonathan Buss who shared a manuscript \cite{Buss:2014} on speed-up results relating time and space.  We also thank Michael Fischer, Mike Paterson, and Nick Pippenger, who tracked down two manuscripts related to circuit simulations.  In addition, we thank Karl Bringmann whose advice helped us to better align this work with recent advances in fine-grained complexity.  Likewise, we recognize helpful discussions with Henning Fernau and all of the participants at the workshop on Modern Aspects of Complexity Within Formal Languages (sponsored by DFG).  Finally, we very much appreciate all of the feedback from Joseph Swernofsky, suggestions from Ryan Williams, and comments from anonymous referees.

\bibliographystyle{plainurl}
\bibliography{references}

\appendix

\end{document}